\documentclass[letterpaper, 10 pt, conference]{ieeeconf}  % Comment this line out if you need a4paper

\IEEEoverridecommandlockouts                              % This command is only needed if 
\usepackage{amsthm}
\usepackage{graphicx} % for pdf, bitmapped graphics files
\usepackage{times} % assumes new font selection scheme installed
\usepackage{amsmath} % assumes amsmath package installed
\usepackage{amssymb}  % assumes amsmath package installed
\usepackage{mathtools} % for \coloneqq
\usepackage[noadjust]{cite}
\usepackage{hyperref}
\usepackage{dsfont}
\usepackage{enumerate}
\usepackage{sidecap}
\usepackage{xcolor}

\usepackage[ruled]{algorithm2e}
\usepackage{algpseudocode}

\usepackage{tikz}
\usetikzlibrary{shapes,arrows,positioning,calc}

\tikzset{
    block/.style = {draw, rectangle, minimum height=1.5cm, minimum width=3cm, thick, align=center},
    small_block/.style = {draw, rectangle, minimum height=0.6cm, minimum width=1.6cm, thick, align=center},
    sum/.style = {draw, circle, inner sep=0pt, minimum size=3mm, thick},
    input/.style = {coordinate},
    output/.style = {coordinate},
    vec_arrow/.style = {->, >=latex', thick},
    vec_no_arrow/.style = {-, >=latex', thick}
}

\newcommand{\R}{\mathbb{R}}
\newcommand{\st}{\text{subject to}}
\newcommand{\minimize}{\text{min}}
\newcommand{\blkdiag}{\operatorname{blkdiag}}
\newcommand{\diag}{\operatorname{diag}}
\newcommand{\Su}{\mathcal{S}_\ub}
\newcommand{\Sx}{\mathcal{S}_\xb}
\newcommand{\Qb}{\bar{Q}}
\newcommand{\Rb}{\bar{R}}
\newcommand{\xh}{x}
\newcommand{\dtlambda}{\dot{\tilde{\lambda}}}
\newcommand{\tlambda}{\tilde{\lambda}}

\newcommand{\tu}{\tilde{u}}
\newcommand{\tp}{^\mathsf{T}}
\newcommand{\xb}{\mathbf{x}}
\newcommand{\ub}{\mathbf{u}}
\theoremstyle{plain}

\newtheorem{remark}{Remark}
\newtheorem{lemma}{Lemma}

\newtheorem{corollary}{Corollary}
\theoremstyle{definition}

\title{\LARGE \bf
Firing Rate Neural Network Implementations of Model Predictive Control
}

\author{Jaidev Gill and Jing Shuang (Lisa) Li% <-this % stops a space
\thanks{J.G. and J.S.L. are with the Department of Electrical Engineering and Computer Science, University of Michigan, Ann Arbor, MI, 48109, USA.
        {\tt\small \{jaidevg, jslisali\}@umich.edu}. 
        J.G. acknowledges the support of the Natural Sciences and Engineering Research Council of Canada (NSERC), CGRSD-609128-2026. }%
}

\begin{document}

\maketitle
\thispagestyle{empty}
\pagestyle{empty}

%%%%%%%%%%%%%%%%%%%%%%%%%%%%%%%%%%%%%%%%%%%%%%%%%%%%%%%%%%%%%%%%%%%%%%%%%%%%%%%%
\begin{abstract}
Human and animal brains perform planning to enable complex movements and behaviors, a process that can be effectively described using model predictive control (MPC).
How could the brain physically implement MPC?
In this work, we translate model predictive controllers into firing rate neural networks, offering insights into the nonlinear neural dynamics that underpin planning. 
We propose a \textit{constructive} method; no training is required. 
This is done first applying the projected gradient method to the dual problem to derive a baseline neural network implementation. We then use factorization and contraction analysis to systematically generate alternative network architectures; in other words, we systematically generate hypotheses for how planning is done in the brain via neural dynamics.
%. This enables us to explore many biologically plausible implementations of MPC. 
Finally, we present numerical simulations to study different neural networks performing MPC to balance an inverted pendulum on a cart (i.e., balancing a stick on a hand), including one example in which imposing sparse connectivity (a property observed in brain networks) does not degrade control performance.

%We show that these networks sustain their control performance under sparse connectivity that is enforced to mirror structural properties of brain networks. 

\end{abstract}

%%%%%%%%%%%%%%%%%%%%%%%%%%%%%%%%%%%%%%%%%%%%%%%%%%%%%%%%%%%%%%%%%%%%%%%%%%%%%%%%
\section{Introduction}
%Humans and animals perform sophisticated goal-oriented behaviors with ease; this capability is believed to rely on online planning in the brain \cite{DIEDRICHSEN_2026, ARIANI_2021, ARIANI_2025}.
%Motor control experiments indicate that increasing the planning window improves task proficiency \cite{ARIANI_2021}. Moreover,
%as planned sequences of movement become more complex, interactions between brain regions increase \cite{ARIANI_2025}; 
%this suggests task complexity is related to the structure of the brain network engaged in planning \cite{DIEDRICHSEN_2026}. 
Goal-oriented behavior in humans and animals relies on online planning in the brain \cite{DIEDRICHSEN_2026, ARIANI_2021, ARIANI_2025}. 
Motor control experiments reveal that longer planning horizons improve performance \cite{ARIANI_2021} and complex tasks require increased inter-regional interactions in brain networks \cite{ARIANI_2025}. This suggests that task complexity is related to the structure of the brain network engaged in planning \cite{DIEDRICHSEN_2026}.
To understand this coupling between task complexity and network engagement, 
we must first formalize how planning behaviors can be implemented neurally. 
Model predictive control (MPC) provides a well-suited theoretical framework to formalize this relationship. 
In MPC, at each cued event a constrained optimal control problem (OCP) is solved over a fixed horizon and the first optimal action is executed.
That is, MPC \emph{plans} ahead to determine the best action to be taken. 
Consequently, mapping MPC onto neural circuits offers a pathway to deciphering the neural mechanisms of biological planning.

In this work, we explore the following question: \textbf{how can MPC be implemented in a neural circuit?}
We focus on circuits (i.e., networks) of neurons governed by rectified linear (ReLU) dynamics; this is a standard model for artificial neural networks and biological firing rate neuron models in neuroscience \cite{CENTORRINO_2023}.
In the context of neuroscientific modeling, proposing a single circuit model is insufficient as the nervous system exhibits widespread structural variability while maintaining consistent functional behavior \cite{MARDER_2011}.
This implies that a comprehensive neural model must account for an equivalence class of networks capable of performing the same computation \cite{GILL_Nonlinear_2025, LI_2025, GILL_Linear_2025}.
Prevailing data-driven modeling techniques can result in behaviorally accurate yet biologically misleading conclusions \cite{BRUNTON_2026} and often neglect alternative network implementations entirely.
Conversely, direct mappings of MPC onto neural dynamics are likely to produce biologically implausible circuitry. 
To address these limitations, we not only constructively map MPC onto neural network dynamics, but also provide a principled tool to systematically generate alternative networks.

To achieve this, we first translate the underlying mathematical optimization into neural dynamics using the continuous time projected gradient method \cite{GOKHALE_2024, BOLTE_2003}. While earlier neural dynamics models solve similar quadratic programs \cite{HOPFIELD_1982, HOPFIELD_1985,CHUA_1984, LIU_2006, BALAVOINE_2012, CENTORRINO_2024_PCN}, our formulation is particularly amenable to analysis that enables us to generate equivalent network implementations. Rather than viewing this framework purely as a computational solver, we leverage it as a powerful tool for determining a rich set of alternative neural implementations of MPC while incorporating biological constraints on structure. 
This opens the door for future analysis of how features of neural network structure relate to properties of the planning task.

The remainder of this paper is organized as follows.  
First, we translate a model predictive controller into a firing rate neural network (Sections \ref{sec: problem setup}, \ref{sec: NN implementation}). 
We then show that equivalence between neural circuits can be deduced via matrix factorization (Section \ref{sec:multilayer}), focusing on sparse connectivity in order to mimic structural properties of brain networks \cite{BULLMORE_2012}. 
Next we consider perturbations to neural circuits, both as a mechanism to generate approximately equivalent circuits and to study the robustness of these neural networks (Section \ref{sec: perturbation}). 
We also briefly discuss alternative formulations to be used when MPC feasibility is not guaranteed (Section \ref{sec: feasibility}).
Lastly, we design a model predictive controller to stabilize an inverted pendulum on a cart; these dynamics mirror balancing a stick on a hand; we validate our theoretical work and explore various networks that perform this task (Section \ref{sec: sims}).  Concluding remarks are provided in Section \ref{sec: conclusion}.

\section{Problem Setup}\label{sec: problem setup}
\paragraph*{Notation} We denote the inner product of two vectors $x,y$ with respect to a matrix $P$ as  $\langle x, y\rangle_P = x\tp P y$. The norm with respect to a matrix is $\|x\|_P = \sqrt{\langle x, x\rangle_P}$. The Frobenius norm of a matrix is denoted by $\| A \|_F$, and the $\ell_0$ norm of a matrix (i.e., the number of non-zero elements) is denoted as $\|A\|_0$. The vectorization of a matrix is denoted as $\operatorname{vec}(A)$. For a vector $x$ we write $[x]$ to denote a diagonal matrix $X$ with elements $[X]_{ii} = x_i$. The spectral abscissa of a matrix $A$ is $\alpha(A) = \max_i \{ \text{Re}(\lambda_i(A))\}$.

We consider the control of a linear system: 
\begin{equation}
    \dot{x} = A_c x + B_c u,
\end{equation}
via a model predictive controller. To do so, we sample the plant periodically at times $\tau$ and solve the following OCP:
\begin{subequations} \label{eq: MPC}
    \begin{alignat}{2}
    &\underset{{\substack{u_0, \dots, u_{N-1}, \\ x_0, \dots, x_N}}}{\minimize} &~  & \tfrac{1}{2}\|x_N\|_P^2+\tfrac{1}{2} \sum_{k=0}^{N-1} \| x_k\|_Q^2 + \|u_k\|_R^2   \\
    &\st ~ & &x_{k+1} = Ax_k + Bu_k, ~ k = 0, \dots, N-1 \\
    && & x_k \in \mathcal{X}, \quad k = 0 , \dots, N \\
    && & u_k \in \mathcal{U}, \quad k = 0, \dots, N-1 \\
    && &x_0 = \xh(\tau). 
\end{alignat}
\end{subequations}
The choice to use linear dynamics and a quadratic cost is inspired by motor control modeling in neuroscience which typically considers the linear-quadratic regulator (LQR) problem \cite{LI_2025}. 
Here $\mathcal{X}$ and $\mathcal{U}$ are polyhedral sets that define the state and control constraints, and $R \succ 0$, $Q \succeq 0$, and $P \succeq 0$ parameterize the cost of the OCP. The linear dynamics modeling the plant are described by $A$ and $B$ and are discretized versions of $A_c$ and $B_c$ respectively. Here $\xh(\tau)$ is the observation of the continuous time plant at time $\tau$. 
The corresponding MPC feedback law is
\begin{equation}
    u_{\textsf{MPC}}(x_0) = u^*_0,
\end{equation}
where $u_0^*$ is the first optimal control action and the control law is a zero-order hold between sampling instances.
By introducing concatenated vectors of the predicted state and control
\begin{equation}
    \xb =\begin{bmatrix}x_1\tp & \cdots &  x_N\tp\end{bmatrix}\tp, \quad
    \ub =\begin{bmatrix}u_0\tp & \cdots &  u_{N-1}\tp\end{bmatrix}\tp,
\end{equation}
and the matrices 
\begin{equation}
    \Sx = \begin{bmatrix}
         A \\ \vdots \\ A^N
    \end{bmatrix},~~
    \Su = \begin{bmatrix}
        B & 0 & \cdots & 0 \\
        AB & \ddots  & \ddots & \vdots \\
        \vdots & \ddots & \ddots & 0\\
        A^{N-1}B &\cdots & \cdots & B
    \end{bmatrix},
\end{equation}
we can write $\xb$ in terms of the initial state $x_0$ and control $\ub$
\begin{equation}\label{eq: x relation}
     \xb = \Sx x_0 + \Su \ub.
\end{equation}
Now introducing  
\begin{equation}
    \Qb = \blkdiag(Q, \dots , Q, P),~~ \Rb = \blkdiag(R, \dots , R),
\end{equation}
we can define 
\begin{equation}
    H = \Su \tp \Qb \Su + \Rb \succ 0, \text{ and } S = \Su\tp \Qb \Sx. 
\end{equation}
The polyhedral constraints can be expressed as $G_\xb \xb \leq g_\xb$ and $G_\ub \ub \leq g_\ub$. Finally, stacking the constraints and using \eqref{eq: x relation} to define 
\begin{equation}
     G = \begin{bmatrix}
        G_\ub  \\ G_\xb \Su
    \end{bmatrix}, \quad  T = \begin{bmatrix}
        0 \\ -G_\xb \Sx
    \end{bmatrix},
    \text{ and } 
    g = \begin{bmatrix} 
    g_\ub \\  g_\xb 
    \end{bmatrix}
\end{equation}
we arrive at the following program  which is equivalent to \eqref{eq: MPC}
\begin{subequations} \label{eq: QP}
    \begin{alignat}{2}
    &\underset{\ub}{\minimize}  &\quad& \tfrac{1}{2}  \ub\tp H \ub + x_0\tp S\tp \ub\\
    &\st  && G\ub \leq g + T x_0.  
\end{alignat}
\end{subequations}
Now that we have simplified the quadratic program, we aim to determine how it could be solved using neural dynamics. 

\section{Neural Network Implementation of MPC}\label{sec: NN implementation}
In order to solve \eqref{eq: QP} with neural dynamics we apply a continuous time optimization method. 
In general the constraints in \eqref{eq: QP} do not admit a closed form projection operator, thus we cannot naturally solve this problem with projected gradient methods. To overcome this, we consider the dual program to \eqref{eq: QP}, which is 
\begin{align}\label{eq: dual}
    \underset{\lambda\geq 0}{\minimize} ~ & \tfrac{1}{2} \lambda \tp GH^{-1}G\tp \lambda + (GH^{-1}S x_0 + g + Tx_0)\tp \lambda.
\end{align}
To solve \eqref{eq: dual} we apply the continuous projected gradient method \cite{BOLTE_2003} which results in the following dynamics:
\begin{equation}\label{eq: nn}
    \eta\dot{\lambda} = - \lambda + \Phi\big( \Gamma\lambda - \alpha (Mx_0 + g) \big),
\end{equation}
where
\begin{equation}
\begin{aligned}
    \Gamma &=I - \alpha GH^{-1}G\tp, \\
     M &= GH^{-1}S + T .
\end{aligned}
\end{equation}

The dynamics \eqref{eq: nn} can be interpreted as a neural network \cite{CENTORRINO_2024_PCN, MILLER_2012} with $m$ neurons, one for each dual variable (i.e., one for each constraint in \eqref{eq: QP}). 
Specifically, the value of the dual variables in \eqref{eq: nn} represent the firing rates of the neurons; alternatively, they can be interpreted in terms of neural membrane voltages \cite{MILLER_2012}. 
The timescale of the network is denoted by $\eta$. 
Here $\Gamma$ is the synaptic weight matrix that encodes the neural circuit's network structure, and
$\alpha$ is an arbitrary step size. The matrix $M$ maps the input measurement $x_0 = \xh(\tau)$ of the state to a bias, and $\Phi(\cdot) = \max\{ \cdot, 0\}$ is a nonlinear activation that is applied element-wise. Here, we leverage the fact that the projection onto the non-negative orthant has a closed form solution given by $\Phi$. 

Since strong duality holds \cite{BOYD__2004}, to recover the minimizer of \eqref{eq: QP} we differentiate the Lagrangian of \eqref{eq: QP} to establish:
\begin{equation}\label{eq: control_sol}
    \ub = -H^{-1}G\tp\lambda - H^{-1} S x_0  .
\end{equation}
Defining $\Upsilon = \begin{bmatrix}
    I & 0 &\cdots & 0 
\end{bmatrix}$ in order to select the first control action from the solution \eqref{eq: control_sol}, the MPC feedback control law is given by: 
\begin{equation}\label{eq: mpc feedback}
    u_{\textsf{MPC}} = -\Upsilon H^{-1}G\tp\lambda - \Upsilon H^{-1} S x_0.
\end{equation}

The overall controller is displayed in Fig.~\ref{fig: overview}.
The control law can be decomposed into the unconstrained solution (via the second term) and the contribution of the constraints through the evolution of the neural network \eqref{eq: nn} (via the first term). In subsequent discussions we drop $\eta$ for simplicity and assume it is sufficiently small such that the neural network can reach equilibrium quickly when compared to the speed at which the state observation $x_0 = \xh(\tau)$ is updated. 

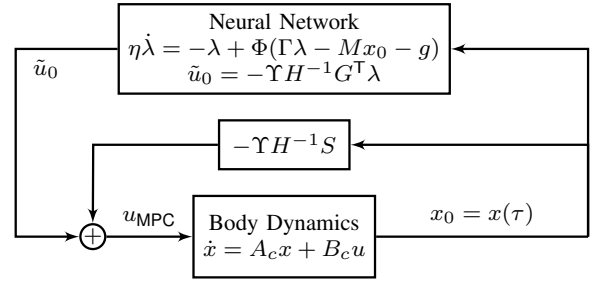
\begin{figure}[!t]
\centering
\resizebox{0.9\columnwidth}{!}{%  <-- This starts the scaling
\begin{tikzpicture}[auto, node distance=0 and 3em, >=latex', scale = 1, every node/.style={transform shape, font=\footnotesize}]

    % --- Place Nodes ---
    \node [sum] (sum) {$+$};
    \node [input, left= of sum] (input_start) {};
    
    \node [block, right=of sum, minimum height=1.0cm, minimum width=1.1cm] (plant_box) {Body Dynamics\\ $\dot \xh =A_c \xh + B_c u $};
    % Place the drawing *inside* the plant box with relative coordinates
    %\pic [scale=0.35] at ($(plant_box.north)!0.85!(plant_box.south)$) {pendulum};

    \node [small_block, above= 0.3cm of plant_box] (inner_controller) {$-\Upsilon H^{-1}S$};
    
    \node [block, above= 0.3cm of inner_controller,minimum height =1 cm,  minimum width=2cm] (neural net) {Neural Network \\ $\eta\dot \lambda = - \lambda + \Phi (\Gamma \lambda - M x_0 - g)$ \\ $\tu_0 = -\Upsilon H^{-1} G\tp\lambda$};
    %\node [small_block, left=of inner_controller] (gain) {$-\Upsilon H^{-1} G\tp$};

    % --- Draw Connections and Labels ---
    \draw [vec_arrow] (sum) -- node[name=plant_input_label] {$u_{\textsf{MPC}}$} (plant_box.west);
    \draw [vec_no_arrow] (plant_box.east) -- node[name=plant_output_label] {$x_0 = \xh(\tau)$} ++(2.6cm,0) coordinate (final_output);
    
    % Main output split point
    \coordinate (split1) at ($(plant_box.east)!1!(final_output)$);
    
    % Draw path to Observer box (large bottom)
    \draw [vec_arrow] (split1) |- (neural net.east);
    
    % Branch from plant output split point 2
    \coordinate (split2) at ($(split1)!0.5!($(plant_box.east)+(3cm,0)$)$); % midpoint of path from split1 to final_output
    % Redefining to match layout: path to small feedback comes from main line.
    \draw [vec_arrow] (split1) |- (inner_controller.east);
    
    % Connect Inner Controller back to sum (this is the key multi-path part)
    % Draw to the *left* of the summing point and up
    \draw [vec_arrow] (inner_controller.west) -- ++(-1.53cm,0) -- (sum.north);
    %\node [anchor=south] at ($(inner_controller.west)+(-0.5cm,0)$) {}; % optional math label

    % Connect Neural net to Gain
    \draw [vec_arrow] (neural net.west) -- ++(-0.5cm,0) -- node[name=out_nn_states_label, below] {$\tu_0$} ++(-0.75cm,0) |- (sum.west);
    
    % Connect Gain back to sum
   % \draw [vec_arrow] (gain.north) -- ++(-0.0cm,0) |- (sum.west);
   % \node [anchor=south] at ($(gain.west)+(-0.25cm,0)$) {}; % optional math label

\end{tikzpicture}
} % <-- This ends the scaling
\caption{Block diagram of a neural network performing MPC. The controller is split into two parts; an offline state feedback law that corresponds to the unconstrained solution (inner loop) and an online neural network that ensures constraint satisfaction (outer loop) (see \eqref{eq: mpc feedback}).}
\label{fig: overview}
\end{figure}

By \cite[Thm. 3.1]{BOLTE_2003} the trajectories of \eqref{eq: nn} converge weakly (i.e., the cost converges) to a minimizer of \eqref{eq: dual} independent of step size, $\alpha$. The choice of step size can be optimized to achieve exponential weak convergence \cite{GOKHALE_2024, HASSANMOGHADDAM_2021}. For simplicity of notation we set $\alpha = 1$. As is typical with MPC, if the number of constraints exceeds the number of optimization variables in \eqref{eq: QP}, the matrix $GH^{-1} G\tp$ may become rank deficient; this implies that there could be potentially many solutions to \eqref{eq: dual}. To overcome this ambiguity, one can adopt a modification of \eqref{eq: nn} proposed in \cite{BOLTE_2003}\footnote{Alternatively, one could assume that the \textit{linear independence constraint qualification} (LICQ) \cite{NOCEDAL_2006} holds; this constraint qualification requires that the gradients of the active constraints at the solution of \eqref{eq: QP} form a linearly independent set.}:
\begin{equation}\label{eq: nn_eps}
    \dot{\lambda} = - \lambda + \Phi\big( \Gamma\lambda - \epsilon(t)\lambda - Mx_0 - g \big),
\end{equation}
where $\epsilon : \R_{\geq0} \rightarrow \R_{\geq0}$ is a non-increasing function that converges to zero. Provided 
\begin{equation}
    \int_0^{\infty} \epsilon(t) dt = \infty ,
\end{equation}
and $\dot{\epsilon}$ is bounded and also converges to zero, then \cite[Thm. 5.1]{BOLTE_2003} guarantees that \eqref{eq: nn_eps} converges strongly (i.e., the dual variables converge) to the minimal norm solution of \eqref{eq: dual}. 
In summary, firing rate neural networks \eqref{eq: nn} and \eqref{eq: nn_eps} both solve the original OCP \eqref{eq: MPC};
\eqref{eq: nn_eps} has the additional feature of ensuring uniqueness of the equilibrium point which can be beneficial in rank-deficient conditions (described above).
We will use both formulations to generate alternative, equivalent networks in subsequent sections.

Immediately, we can observe some interesting properties of the network structure of the synaptic weights $\Gamma$. If the constraints encoded in $G$ couple many of the control actions explicitly through $G_\ub$ or through the influence of the dynamics $G_\xb\Su$ then $G$ will be dense and thus $\Gamma$ is most likely dense. Moreover, if the constraints are relaxed, then
\begin{equation}
    \Phi(\Gamma \lambda -Mx_0 - g ) \rightarrow 0 \text{~~as~~} g \rightarrow \infty
\end{equation}
and consequently $\lambda \rightarrow 0$ as the neural network evolves. Then the constraints are always inactive and the network structure plays no role in computing the control action and we recover the unconstrained solution to the OCP \eqref{eq: MPC}: 
\begin{equation}
    \ub = - H^{-1} S x_0.
\end{equation}
This indicates that the network plays a crucial role when constraints are active over the prediction horizon. 

\section{Multilayer Equivalent Networks}\label{sec:multilayer}

The firing rate neural networks \eqref{eq: nn} and \eqref{eq: nn_eps} suggest two different ways to implement MPC using neural dynamics. 
%determine the MPC control law through neural dynamics, i.e., via the projected gradient method applied to the dual problem \eqref{eq: dual}. 
In fact, previous work (e.g., \cite{LIU_2006} and the references therein) has explored a variety of methods to translate quadratic programs such as \eqref{eq: QP} into neural networks. 
However, our motivating application is neuroscience, and it is highly unlikely that the brain directly implements neural networks whose firing rates happen to directly represent primal or dual variables in an optimization problem, as is the case with nearly all neural networks generated through conventional means (including \eqref{eq: nn} and \eqref{eq: nn_eps}).
Biology only requires that the neural network recovers the solution to \eqref{eq: MPC}. 
%To address this, we explore how we can systematically generate alternative networks.
In fact, there can be many possible neural networks that perform the same behavior (e.g., MPC); this is a common observation in computational neuroscience \cite{MARDER_2011, LI_2025, GILL_Linear_2025, GILL_Nonlinear_2025}.
In this section and subsequent sections, we formalize this observation to the problem at hand and systematically generate alternative neural networks that implement \eqref{eq: QP}.
We start by extending the finding from \cite{MILLER_2012} that the synaptic weight matrix $\Gamma$ can be chosen to be inside or outside the nonlinear activation function with an appropriate change to the input. %Extending this idea, we show that there is a continuum of neural networks that are equivalent.  

Consider the MPC feedback law implemented by the firing rate network in \eqref{eq: nn}, i.e., 
\begin{equation}\label{sys: nn_before_Lure}
    \begin{cases}
        \dot{\lambda} = - \lambda + \Phi(\Gamma \lambda - M x_0 - g), \\
        u_{\textsf{MPC}} = -\Upsilon H^{-1}G\tp \lambda - \Upsilon H^{-1} S x_0 .
    \end{cases}
\end{equation}
We can rewrite \eqref{sys: nn_before_Lure} as a Lur'e system:
\begin{equation}\label{sys: Lur'e}
    \begin{cases}
        \dot \lambda = - \lambda + \nu, \\
        y = \Gamma \lambda, \\
        \tu_0 = - \Upsilon H^{-1} G\tp \lambda, \\
        \nu = \Phi(y - Mx_0 -g ), \\
        u_{\textsf{MPC}} = \tu_0 - \Upsilon H^{-1} S x_0. 
    \end{cases}
\end{equation}
With this rewrite, we can establish the following lemma. 
\begin{lemma}\label{lem: equivalent nn}
    The system \eqref{sys: nn_before_Lure} can be rewritten as a multilayer network of the form 
    \begin{equation}\label{sys: multilayer}
    \begin{cases}
        \dtlambda = - \tlambda + \Psi\Phi(\Omega_1 \tlambda - M x_0 - g), \\
        u_{\textnormal{\textsf{MPC}}} = \Omega_2 \tlambda - \Upsilon H^{-1} S x_0 ,
    \end{cases}
    \end{equation}
    where 
    \begin{equation}\label{eq: factorization}
        \begin{bmatrix}
            \Gamma \\ -\Upsilon H^{-1}G\tp
        \end{bmatrix} = \begin{bmatrix}
            \Omega_1 \\ \Omega_2 
        \end{bmatrix}\Psi .
    \end{equation}
\end{lemma}
\begin{proof}
    Consider the linear part of the Lur'e system in \eqref{sys: Lur'e}:
    \begin{equation}
    \begin{cases}
        \dot \lambda = - \lambda + \nu, \\
        y = \Gamma \lambda, \\
        \tu_0 = - \Upsilon H^{-1} G\tp \lambda. \\
    \end{cases}
    \end{equation}
    Notice that the transfer function from the input $\nu$ to the outputs $(y, \tu_0)$ is 
    \begin{equation}
        \frac{1}{s+1} \begin{bmatrix}
            \Gamma \\ -\Upsilon H^{-1} G\tp 
        \end{bmatrix},
    \end{equation}
    thus any factorization of the form 
    \begin{equation}
        \begin{bmatrix}
            \Gamma \\ -\Upsilon H^{-1}G\tp 
        \end{bmatrix} = \begin{bmatrix}
            \Omega_1 \\ \Omega_2 
        \end{bmatrix}\Psi ,
    \end{equation}
    leads to a realization of the linear part of the Lur'e system:
    \begin{equation}
    \begin{cases}
        \dot \tlambda = - \tlambda + \Psi\nu, \\
        y = \Omega_1 \tlambda, \\
        \tu_0 = \Omega_2 \tlambda. \\
    \end{cases}
    \end{equation}
    Substituting in for $\nu$ leads to \eqref{sys: multilayer}. 
\end{proof}
We can interpret $\Omega_1$ and $\Psi$ in \eqref{sys: multilayer} as two layers of a new neural network. 
Note that $\tlambda$ may be negative; thus, this network can be interpreted in terms of neural membrane voltages or firing rates relative to a baseline (if $\tlambda$ is lower bounded), but not raw firing rates. 
The results of Lemma \ref{lem: equivalent nn} indicate that any factorization \eqref{eq: factorization} leads to an equivalent network representation. This motivates the study of nonconvex optimization programs of the form 
\begin{equation}\label{eq: factorization general}
    \underset{\Omega_1, \Omega_2, \Psi}{\minimize} \quad \Bigg\| \begin{bmatrix}
            \Gamma \\ -\Upsilon H^{-1}G\tp 
        \end{bmatrix} - \begin{bmatrix}
            \Omega_1 \\ \Omega_2 
        \end{bmatrix}\Psi \Bigg\|_F^2,
\end{equation}
subject to constraints that lead to desirable network features. Similar programs have been extensively studied in optimization \cite{BOLTE_2014} and signal processing \cite{BAO_2014} literature. 
Since the connection structure of the brain is sparse \cite{NEWMAN_2018, BULLMORE_2012}, we would like to incorporate sparsity into the neural networks. To do so, we consider the following program: 
\begin{subequations}\label{eq: factor matrix}
\begin{alignat}{2}
    &\underset{\Omega, \Psi}{\minimize} &\quad &\| \Theta  - \Omega\Psi \|_F^2 \\
    &\st && \|\Omega\| _0 \leq s_\Omega \\
    &    && \| \Psi\|_0 \leq s_{\Psi}.
\end{alignat}
\end{subequations}

Here $\Omega = [
    \Omega_1 \tp , ~\Omega_2 \tp 
]\tp $ and $\Theta = [
            \Gamma\tp , ~(-\Upsilon H^{-1}G\tp )\tp
        ]\tp$.
This program attempts a faithful factorization of the matrix $\Theta$ while enforcing sparsity constraints on the factors. The sparsity levels $(s_\Omega, s_\Psi)$ encode the maximum number of non-zero elements to be kept in each of the factors. To solve this problem, we use the \textit{proximal alternating linearized minimization} (PALM) algorithm developed in \cite{BOLTE_2014}; it is displayed for completeness in Algorithm \ref{alg: palm}. Note that $\Omega_0, \Psi_0$ are initial guesses for the factors that can be chosen freely, $\beta_1, \beta_2$ are step sizes, and $\mathcal{P}_s(\cdot)$ leaves the $s$ largest (in absolute value) entries of the argument unchanged and sets all others to zero. Convergence guarantees for the PALM algorithm can be found in \cite{BOLTE_2014, BAO_2014}.
Apart from optimizing over both matrix factors, by leaving one of the factors fixed \eqref{eq: factorization general} becomes a convex program provided the constraints that are incorporated are convex. Notably the unconstrained problem will admit solutions:
\begin{equation}
    \Omega^* = \Theta \Psi\tp(\Psi\Psi\tp)^{-1} 
\end{equation}
if $\Psi$ is fixed and, 
\begin{equation}
    \Psi^* = (\Omega \tp \Omega )^{-1} \Omega\tp \Theta 
\end{equation}
if $\Omega$ is fixed. This enables direct enforcement of a particular structure onto the factor of choice. 

\begin{algorithm}\label{alg: palm}
\caption{PALM \cite{BOLTE_2014} to solve \eqref{eq: factor matrix}}
    \KwData{$\Theta$, $\Omega_0$, $\Psi_0$, $\beta_1$, $\beta_2$, $s_{\Omega}$, $s_{\Psi}$}
    \For{$k = 0$ \KwTo $\bar{k}$}{
    $\tilde{\Omega}_{k} \leftarrow \Omega_k - \tfrac{1}{\beta_1 \|\Psi_k\Psi_k\tp \|_F}\big( \Omega_k \Psi_k - \Theta\big)\Psi_k\tp$\;
    $\Omega_{k+1} \leftarrow \mathcal{P}_{s_\Omega}(\tilde{\Omega}_{k})$\;

    $\tilde{\Psi}_{k} \leftarrow \Psi_k - \tfrac{1}{\beta_2 \|\Omega_{k+1}\tp\Omega_{k+1}\ \|_F}\Omega_{k+1}\tp\big( \Omega_{k+1} \Psi_k - \Theta\big)$\;
    $\Psi_{k+1} \leftarrow \mathcal{P}_{s_\Psi}(\tilde{\Psi}_{k})$\;
    }
\end{algorithm}

Having discussed exact equivalence among neural network representations, we now explore a non-exact method to generate alternative neural networks.
%hypotheses for neural circuit structure.

\section{Perturbed Networks \& Approximate Equivalence}\label{sec: perturbation} 
%As previously mentioned, the neural networks \eqref{eq: nn} and \eqref{eq: nn_eps} are one representation of the neural dynamics performing MPC, 
In the previous section, we generated multilayer networks that result in \emph{identical} control actions.
In reality, biological neural networks may also \emph{approximate} a true model predictive controller as opposed to implementing it exactly. 
This motivates understanding how the network structure of a single layer firing rate network \eqref{eq: nn} or \eqref{eq: nn_eps} 
could be modified while approximately maintaining its output behavior. 
By exploring this question, we can try to quantify structural importance of edges (i.e., the \emph{robustness} of the network). 
Specifically, we bound the difference between the control actions determined by two different firing rate networks \eqref{eq: nn_eps}. For simplicity we consider  \eqref{eq: nn_eps}, and the results that follow also trivially apply to networks of the form given in \eqref{eq: nn}. 

We consider how the control actions produced by networks of the form \eqref{eq: nn_eps}:
\begin{equation}\label{sys: orig}
\begin{cases}
        \dot{\lambda}_1 = -\lambda_1 + \Phi \big(\Gamma \lambda_1 - \epsilon(t)\lambda_1 - Mx_0^1 - g\big), \\
        u_{\textsf{MPC}}^1 = -\Upsilon H^{-1}G\tp \lambda_1  -\Upsilon H^{-1}Sx_0^1 ,
\end{cases}
\end{equation}
and
\begin{equation}\label{sys: perturbed}
\begin{cases}
        \dot{\lambda}_2 = -\lambda_2 + \Phi \big((\Gamma+\Delta) \lambda_2 - \epsilon(t)\lambda_2 - Mx_0^2 - g\big), \\
        u_{\textsf{MPC}}^2 = -\Upsilon H^{-1}G\tp \lambda_2  -\Upsilon H^{-1}Sx_0^2 ,
\end{cases}
\end{equation}
are changed due to wiring differences $\Delta$ in the synaptic weight matrix and differences in the current state $x_0 = \xh(\tau)$. Recall, the timescale at which the neural network is operating is assumed to be much faster than the rate at which $x_0 = x(\tau)$ is updated, and all analysis is performed with respect to this faster timescale. 
\begin{lemma}\label{lem:diff_eq}
    Suppose that the dynamics of the perturbed system \eqref{sys: perturbed} are contracting, i.e.,  $\exists \mathsf{P} \succ 0$ and $\mu < 1$ such that $\forall \lambda_1, \lambda_2, b$ and $\forall t$
    \begin{multline}\label{eq: osLip}
            \big\langle \Phi\big((\Gamma + \Delta - \epsilon I) \lambda_1  +b \big) - \Phi\big((\Gamma + \Delta - \epsilon I) \lambda_2  +b\big),\\ \lambda_1 - \lambda_2 \big\rangle_\mathsf{P} \leq \mu \|\lambda_1 - \lambda_2 \| ^2_\mathsf{P},
    \end{multline}
    then the difference in the dual variables $\lambda_1$ and $\lambda_2$ follows the differential inequality 
    \begin{multline}
        \tfrac{d}{dt} \| \lambda_1 - \lambda_2 \|_\mathsf{P} \leq -(1-\mu)\|\lambda_1 - \lambda_2 \|_\mathsf{P} \\+ \max_{t' \in [0,t], [d] \in [0,1]^m} \big\|[d]M(x_0^1-x_0^2) - [d]\Delta \lambda_1\big\|_\mathsf{P}.
    \end{multline}
\end{lemma}
\begin{proof}
    Identify that 
    \begin{equation}
        \tfrac{d}{dt} \| \lambda_1 - \lambda_2 \|_\mathsf{P} = \tfrac{d}{dt} \sqrt{\| \lambda_1 - \lambda_2 \|_\mathsf{P}^2} = \frac{\frac{d}{dt} \| \lambda_1 - \lambda_2 \|_\mathsf{P}^2}{2 \| \lambda_1 - \lambda_2 \|_\mathsf{P}}.
    \end{equation}
    Expanding the numerator:
    \begin{multline}
        \tfrac{d}{dt} \|\lambda_1 - \lambda_2 \|_\mathsf{P}^2 = -2\big(\|\lambda_1 - \lambda_2\|_\mathsf{P}^2  - \\ \big\langle\Phi\big((\Gamma - \epsilon I)\lambda_1 + b_1\big) - \Phi\big((\Gamma +\Delta - \epsilon I)\lambda_2 + b_2\big)  ,\lambda_1 - \lambda_2 \big\rangle_\mathsf{P} \big),
    \end{multline}
    where $b_1 = - Mx_0^1 - g $ and $b_2 = - Mx_0^2 - g$.
    Adding and subtracting $\Phi\big((\Gamma +\Delta - \epsilon I)\lambda_1 + b_2\big)$ within the inner product, and applying the contractivity assumption we have 
    \begin{multline}
        \tfrac{d}{dt} \|\lambda_1 - \lambda_2 \|_\mathsf{P}^2 \leq -2\big((1-\mu)\|\lambda_1 - \lambda_2\|_\mathsf{P}^2  - \\ \big\langle\Phi\big((\Gamma - \epsilon I)\lambda_1 + b_1\big) - \Phi\big((\Gamma +\Delta - \epsilon I)\lambda_1 + b_2\big)  ,\lambda_1 - \lambda_2 \big\rangle_\mathsf{P}\big) .
    \end{multline}
    Applying the Cauchy-Schwarz inequality we get 
    \begin{multline}
        \tfrac{d}{dt} \|\lambda_1 - \lambda_2 \|_\mathsf{P}^2 \leq -2\big((1-\mu)\|\lambda_1 - \lambda_2\|_\mathsf{P}^2  - \\ \|[d](-\Delta\lambda_1 + (b_1-b_2)) \|_\mathsf{P} \|\lambda_1 - \lambda_2\|_\mathsf{P} \big) .
    \end{multline}
    Here we identified that
    $\Phi\big((\Gamma - \epsilon I)\lambda_1 + b_1\big) - \Phi\big((\Gamma +\Delta - \epsilon I)\lambda_1 + b_2\big) = [d](-\Delta\lambda_1 + (b_1-b_2))$, where $[d]_{ii}$ is the slope between the $i$-th components. Since $\Phi$ is slope restricted between $[0,1]$ we can maximizing over all possible slopes $[d] \in [0, 1]^m$ and time $t' \in [0, t]$ to arrive at the result. 
\end{proof}

The assumption in Lemma \ref{lem:diff_eq} is a one-sided Lipschitz inequality that is equivalent to the Jacobian-based linear matrix inequality for contraction \cite{CENTORRINO_2023}. In \cite{CENTORRINO_2023}, $\mathsf{P}$ and $\mu$ are explicitly determined in terms of the spectra of the synaptic weights encoded in $\Gamma + \Delta$ for the formulation \eqref{eq: nn}. If $\Delta$ is symmetric, contractivity is guaranteed if $\alpha(\Gamma + \Delta) < 1$ \cite{CENTORRINO_2023}. 

\begin{corollary}\label{cor: time_diff}
The difference between the dual variables of \eqref{sys: orig} and \eqref{sys: perturbed} are bounded, specifically they obey the following inequality:  
\begin{multline}
        \| \lambda_1(t) - \lambda_2(t)\|_\mathsf{P} \leq e^{-(1-\mu)t} \| \lambda_1(0) - \lambda_2(0)\|_\mathsf{P} \\+ \tfrac{(1-e^{-(1-\mu)t})}{1-\mu} \max_{t' \in [0,t], [d] \in [0,1]^m} \big\|[d]M(x_0^1-x_0^2) \\- [d]\Delta \lambda_1\big\|_\mathsf{P}.
\end{multline}
\end{corollary}
\begin{proof}
    Apply the comparison principle to the result of Lemma \ref{lem:diff_eq}. 
\end{proof}

\begin{remark}
    Robustness to perturbations in the edge weights relies on the contractivity of the perturbed system. Similar analysis has been performed to study the input-to-state stability of contractive systems, e.g.,  \cite[Thm. 37]{DAVYDOV_2022}.
\end{remark}

% ---- Define Figure here so it shows up on next page 
\begin{figure*}[t]
    \centering
    \includegraphics{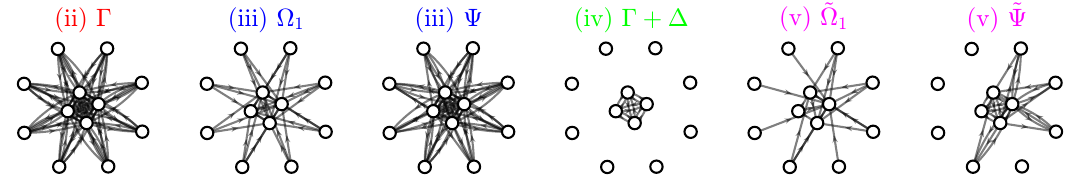}
    \caption{Structures of neural networks implementing MPC. $\Gamma$ corresponds to \eqref{eq: nn}, $\Omega_1$ and $\Psi$ correspond to \eqref{sys: multilayer}. $\Gamma + \Delta$ corresponds to \eqref{sys: perturbed}, $
    \tilde{\Omega}_1$ and $\tilde \Psi$ correspond to a multilayer network where the layers only approximately satisfy \eqref{eq: factorization}. Color of text corresponds to line color in Fig.~\ref{fig:comparison of trajectories}. An edge is deemed present in the network if its weight is greater than $10^{-5}$. Self loops are not displayed.}
    \label{fig:comparison of networks}
\end{figure*}

\begin{corollary}\label{cor: control_bd}
    The difference between control actions computed by \eqref{sys: orig} and \eqref{sys: perturbed} is upper bounded, i.e., 
    \begin{multline}
         \|u_{\textnormal{\textsf{MPC}}}^1 - u_{\textnormal{\textsf{MPC}}}^2 \|_\mathsf{P} \leq \| \Upsilon H^{-1} S(x_0^1-x_0^2)\|_\mathsf{P} \\+ \tfrac{\|\Upsilon H^{-1} G\tp\|_\mathsf{P}}{1-\mu} \max_{t' \in [0,\infty), [d] \in [0,1]^m} \big\|[d]M(x_0^1-x_0^2) \\- [d]\Delta \lambda_1\big\|_\mathsf{P}.
    \end{multline}
\end{corollary}
\begin{proof}
    Apply the triangle inequality to $\|u_{\textsf{MPC}}^1 - u_{\textsf{MPC}}^2 \|_\mathsf{P}$, then apply the result of Corollary \ref{cor: time_diff} to arrive at the result. 
\end{proof}

The bound in Corollary 2 can be tightened by introducing the notion of \emph{contraction in the observable space} discussed in \cite{GILL_Nonlinear_2025}. However, the corresponding one-sided Lipschitz inequalities are typically challenging to satisfy as they often apply to a restricted class of synaptic weights $\Gamma$. 

Corollary \ref{cor: control_bd} can be used to determine how robust the network is to perturbations to the edge weights. For instance, if we have prior knowledge that certain dual variables are small, then $\Delta$ can be chosen appropriately to remove corresponding edges and generate a new neural network. The caveat with this procedure is that we must know the trajectories $\lambda_1(t)$ in advance. 
A similar problem to targeted edge removal is to develop 
neural circuit redesign programs that attempt to maintain the behavior of the controller. For example, restricting to symmetric perturbations, consider the following optimization program: 
\begin{subequations}
    \begin{alignat}{2}
    &\underset{\Delta}{\minimize} &\quad & \|\operatorname{vec}(\Gamma + \Delta) \|_1 \label{eq: obj} \\
    &\st && \eqref{eq: osLip} \label{eq: net_contr}\\
    & && \| \Delta \|_\mathsf{P} \leq \gamma \label{eq: net_tolerance}.
\end{alignat}
\end{subequations}

The objective \eqref{eq: obj} encourages the neural network (i.e., synaptic weight matrix $\Gamma + \Delta$ to be sparse. The constraint \eqref{eq: net_contr} ensures the network is contracting, as previously mentioned when considering the formulation \eqref{eq: nn} this constraint can be equivalently written as $\alpha(\Gamma + \Delta) < 1$ \cite{CENTORRINO_2023}.  Lastly, \eqref{eq: net_tolerance} enforces a desired tolerance ($\gamma$) up to which the network can be perturbed.

\section{Ensuring Feasibility}\label{sec: feasibility}
The previous sections discussed different implementations of the same model predictive controller. In addition, it is valuable to consider how the OCP formulation \eqref{eq: MPC} can be changed and hence understand the influence of these changes on the network structure. 
Here, we focus on changes that ensure feasibility of the OCP \eqref{eq: MPC}.
In known environments, organisms typically will have a strong understanding of their surroundings (i.e., constraints, dynamics, etc.). In this case, the formulation \eqref{eq: MPC} is a suitable model of planning. However, in unknown environments, this assumption is violated. This could lead to infeasibility of the program \eqref{eq: MPC} and thus unexpected behavior of the neural network \eqref{eq: nn}. 
One approach to overcome this is to employ slack variables to to relax constraints in \eqref{eq: QP}\footnote{Other potential approaches include moving the constraints into the objective, etc. We defer this to future investigation.}. 
We point the reader to \cite{KERRIGAN_2000} for a thorough discussion on the choice of penalties for the slack variable formulations. Specifically, we consider the following quadratic program:
\begin{subequations} \label{eq: QP_slack}
    \begin{alignat}{2}
    &\underset{\ub, s}{\minimize}  &\quad& \tfrac{1}{2}  \ub\tp H \ub + x_0\tp S\tp \ub + \sigma s\tp s \\
    &\st && G\ub \leq g + T x_0 + Es \\
    & && s \geq 0.  
\end{alignat}
\end{subequations}
Here $\sigma$ is a penalty parameter and $E = \begin{bmatrix}
        0 & I
    \end{bmatrix}\tp$
is used to relax the state constraints; the control constraints can always be enforced. 
With this modification, the resulting network is 
\begin{equation}\label{eq: slack_network}
    \begin{bmatrix}
        \Gamma - {\sigma}^{-1}EE\tp &  -{\sigma}^{-1}E \\
        -{\sigma}^{-1}E\tp & (1- {\sigma}^{-1})I
    \end{bmatrix}.
\end{equation}
Since $EE\tp $ is diagonal and the off diagonal blocks simply couple state constraint nodes to new slack variable nodes, the core structure of the network remains (see Fig.~\ref{fig:slack_network}). 
This is yet another neural network implementation of MPC.
Furthermore, techniques previously discussed equivalently apply to this  formulation.

\section{Simulations}\label{sec: sims}

\subsection{Alternative Network Implementations}
Using the previously described methods, we construct several neural circuits that implement MPC.
Consider a pendulum on a cart with nonlinear dynamics:
\begin{equation}
    \begin{cases}
        (M+m) \ddot{y} + m\ell\ddot{\theta} \cos \theta - m \ell \dot{\theta}^2\sin \theta = u, \\
        \ell \ddot{\theta} + \ddot{y}\cos{\theta} - g \sin \theta = 0,
    \end{cases}
\end{equation}
here $M= 0.5$ and $m =0.4$ are the masses of the cart and pendulum respectively, $\ell = 1$ is the length of the pendulum, $g = 9.81$ is the gravitational acceleration, $\theta$ is the angle from the upright position, and $y$ is the position of the cart. Defining the state vector $\xh = \begin{bmatrix}
    y& \dot{y}& \theta& \dot{\theta}
\end{bmatrix}\tp$  
and linearizing about the upright position the continuous time linear dynamics are:
\begin{equation}\label{eq: lin pend dyn}
    \dot \xh = \begin{bmatrix}
        0 & 1 & 0 & 0 \\
        0 & 0 & \tfrac{-m g}{M} & 0 \\
        0 & 0 & 0 & 1 \\
        0 & 0 & \tfrac{(M+m)g}{M\ell} & 0 
    \end{bmatrix}\xh + \begin{bmatrix}
        0 \\ \tfrac{1}{M} \\ 0 \\ \tfrac{-1}{M\ell}
    \end{bmatrix} u. 
\end{equation}
Applying a zero-order hold discretization with a sampling rate $T_s = 0.02$, we design a MPC feedback law by choosing of $Q = \diag([10, 1, 500, 1])$,  $R = 0.1I$, and $P$ a solution to the discrete algebraic Ricatti equation. This choice of $P$ ensures that when the constraints are inactive the solution of \eqref{eq: MPC} will be equivalent to the infinite horizon linear quadratic regulator (LQR) problem.
We enforce the state constraints: 
\begin{equation}
    \begin{bmatrix}
        -0.62\\  -0.1
    \end{bmatrix} \leq Cx_k \leq \begin{bmatrix}
        0.62 \\ 1
    \end{bmatrix} , ~~  C = \begin{bmatrix} 1 & 0 & 0 & 0 \\ 0 & 0 & 1 & 0 
\end{bmatrix},
\end{equation}
where $C$ selects the cart position and angle of the pendulum. The control constraints are $-10\leq u_k \leq 12$.
With these choices, the system admits trajectories that lead to frequent occurrences of active constraints; this allows us to better validate theoretical results (see end of Section \ref{sec: NN implementation}).
We choose a horizon $N=2$. 
We then determine the corresponding neural network \eqref{eq: nn} following Section \ref{sec: NN implementation} and determine the synaptic weight matrix $\Gamma$. 
We factorize $\Gamma$ via Algorithm \ref{alg: palm} with $s_\Omega = 144$ and $s_\Psi = 144$. 
In addition, we consider an approximated factorization (using methods described in Section \ref{sec:multilayer}), $\Theta \approx \tilde \Omega \tilde \Psi$ by setting $s_\Omega = 40$ and $s_\Psi = 40$. 
We also consider the trajectories of a perturbed network $\Gamma + \Delta$ where we have removed all edge weights less than $0.01$ (in absolute value) and then subtracted $10^{-4}I$ from the resulting network to ensure it is contracting. 

In Fig.~\ref{fig:comparison of networks} we display the original network $\Gamma$, the multilayer network defined by $\Omega_1$ and $\Psi$, the perturbed network $\Gamma +  \Delta$, and lastly the approximate multilayer network defined by $\tilde \Omega_1$  and $\tilde \Psi$. 
The layers of the multilayer network encoded in $\Omega_1$ and $\Psi$ mirror $\Gamma$ but appear to have a less dense set of connections.  
The approximate factorization given by $\tilde \Omega_1$ and $ \tilde \Psi$ illustrate how sparse the network can be while still maintaining the controller behavior.

In Fig.~\ref{fig:comparison of trajectories} we display the resulting trajectories of model predictive controllers that determine the control action via (i) a traditional quadratic programming solver (ii) the firing rate formulation \eqref{eq: nn} (iii) the factorized multilayer network formulation \eqref{sys: multilayer} (iv) the perturbed network using the formulation in \eqref{eq: nn} and (v) the approximate multilayer network. 
Between each sample, we evolve the state of the system via \eqref{eq: lin pend dyn}. 
We see that three implementations (i-iii) result in near-exact equivalence and display trajectories that are visually indistinguishable. 
This validates the previous discussions on equivalence of the three formulations. 
In contrast, the perturbed network (iv) deviates from the other trajectories and does not satisfy the state constraints; this is to be expected as this network is only an approximate representation of the controller. We note that the trajectories of (iv) are broadly similar otherwise. 
Notably, through the approximate factorization (v), we have generated a new hypothesis network that closely tracks the true MPC solution despite drastic changes to the structure. Specifically, we are able to make the network layers sparse to mimic brain connectivity while still maintaining controller behavior.

% The perturbed network $\Gamma + \Delta$ removes all influence from the external nodes that are arranged in a circle around the four nodes in the center of the network\footnote{This does not imply that the dual variables that correspond to the external nodes will tend to zero.}. These four central nodes correspond to the control constraints.
% Despite the removed edges being weak, the simulation results indicate that these edges are critical for constraint satisfaction. Hence, we can rule out $\Gamma + \Delta$ as a hypothesis neural circuit. 

\begin{figure}
    \centering
    \includegraphics{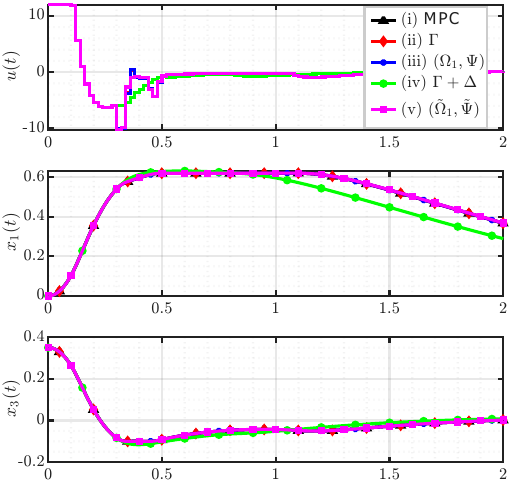}
    \caption{Control action ($u$) and trajectories of the cart position ($\xh_1$) and pendulum angle ($\xh_3$). \textsf{MPC} corresponds to the use of a traditional optimizer. The remaining lines correspond to the neural network implementations defined by the structures given in Fig.~\ref{fig:comparison of networks}. Color of lines correspond to text color in Fig.~\ref{fig:comparison of networks}.}
    \label{fig:comparison of trajectories}
\end{figure}

% To compare the original network to the multilayer network in further detail, we examine the degree distributions which can be found in Fig.~\ref{fig:comparison deg distrib}. In the original network $\Gamma$, the degree distributions exhibit a bimodal structure. After factorization the in-degree (i.e., the number of incoming edges) distribution of $\Omega_1$ is shifted towards zero, with only the four central nodes having incoming connections, while the distribution for $\Psi$ appears similar to $\Gamma$ with some of the incoming edges removed.  
% The out-degree (i.e., the number of outgoing edges) distribution of $\Omega_1$ is shifted entirely below $k = 5$ and $\Psi$ also appears to have a reduced number of outgoing connections.  

% \begin{figure}
%     \centering
%     \includegraphics[width=\linewidth]{Figures/degree_distributions.pdf}
%     \caption{In-degree and out-degree distributions of the original network structure ($\Gamma$) and multilayer network ($\Omega_1$ \& $\Psi$).}
%     \label{fig:comparison deg distrib}
% \end{figure}

\begin{figure}
    \centering
    \includegraphics{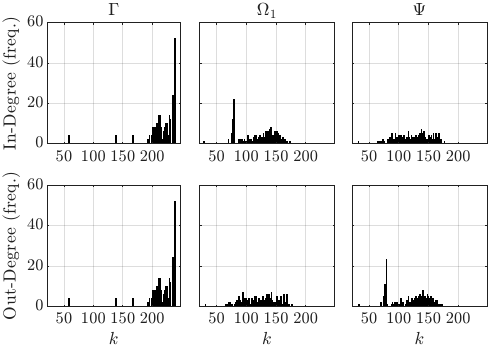}
    \caption{In-degree and out-degree distributions of the original network structure ($\Gamma$) and multilayer network ($\Omega_1$ \& $\Psi$) for the longer planning horizon $N = 40$.}
    \label{fig:comparison deg distrib large scale}
\end{figure}
\subsection{Extending the Horizon}
In the previous simulations we had used a smaller network in the interest of clearer visualizations, we now explore larger more brain-like networks. 
To do so we extended the planning horizon to $N= 40$ which corresponds to networks with 240 nodes. In Fig.~\ref{fig:comparison deg distrib large scale} we display the degree distribution for the original $\Gamma$, and the factorized networks $\Omega_1$ and $\Psi$. The relative reconstruction error after applying the PALM algorithm was $1.18\%$, and the density of the networks, as computed by the ratio of the number of edges over the maximum number of edges, was initially $0.91$ whereas the factorized networks had densities of $0.52$, a near factor of two reduction in edges. In Fig.~\ref{fig:comparison deg distrib large scale} the original network's $\Gamma$ degree distributions implies a structure of nearly all-to-all connectivity. After factorization the in-degree (i.e., the number of incoming edges) distribution of $\Omega_1$ and $\Psi$ are shifted left, with all nodes exhibiting a significant reduction in degree. 
The out-degree (i.e., the number of outgoing edges) distributions display similar properties. 
In addition, in Fig.~\ref{fig:small-world} we fixed the structure of $\Psi$ to a small-world Watts-Strogatz network where each node is connected to $4$ of its neighbors and the network is rewired with an edge flip probability of $0.1$. In this case an unconstrained fitting was done to determine $\Omega$ to achieve a cost on the order of $10^{-20}$. This illustrates that characteristics like small-world structure that are ubiquitous in brain networks \cite{BULLMORE_2012, NEWMAN_2018} can be enforced in the network. 

\begin{figure}
    \centering
    \includegraphics{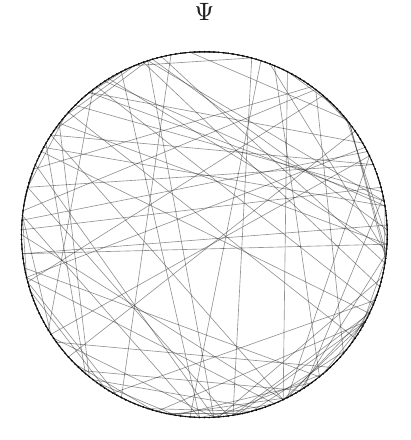}
    \caption{$\Psi$ with small-world structure enforced for the longer planning horizon $N = 40$. Here the 240 nodes form a circle where the neighboring connections make the perimeter appear dense while the sparse long range connections make the network small-world.}
    \label{fig:small-world}
\end{figure}

\subsection{Slack Variable Network}
Lastly, we display the network corresponding to the slack variable formulation \eqref{eq: QP_slack} in Fig.~\ref{fig:slack_network}. As described in \eqref{eq: slack_network}, the addition of the slack variables creates a set of nodes that are connected to the state constraint nodes, and the core structure of the network appears similar to $\Gamma$ in Fig.~\ref{fig:comparison of networks}. The exact same techniques described in Sections \ref{sec:multilayer} and \ref{sec: perturbation} can be used to generate alternative network representations. We have observed that for sufficiently large $\sigma$ this network does indeed mimic the behavior of (i-iii) displayed in Fig.~\ref{fig:comparison of trajectories}.

\begin{figure}[htbp]
  \centering
  \begin{minipage}{0.3\linewidth}
    \centering
    \includegraphics[width=\linewidth]{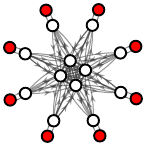}
  \end{minipage}
  \hfill 
  \begin{minipage}{0.65\linewidth}
    \caption{The synaptic weight matrix corresponding to \eqref{eq: QP_slack}. The core structure appears visually identical to the formulation without slack variables. Slack variable constraints are represented by the red nodes.}
    \label{fig:slack_network}
  \end{minipage}
\end{figure}

\section{Concluding Remarks}\label{sec: conclusion}
In this work, we proposed firing rate neural networks to perform planning. 
This work provides a first step towards understanding how planning can be performed via neural circuits. In addition, we offer systematic techniques to generate alternative implementations in both an exact and approximate manner. We also show that it is possible to impose commonly known brain network properties (sparse and small-world connectivity) using our framework.

Future directions of investigation include:
\begin{enumerate}
    \item Extension to control of nonlinear systems via feedback linearization and Koopman lifting.
    \item Analysis of the ensemble of equivalent networks to deduce structural properties that are \emph{necessary} versus incidental. 
\end{enumerate}

%\addtolength{\textheight}{-4cm}   % This command serves to balance the column lengths 


                                  % I COMMENTED IT OUT

%%%%%%%%%%%%%%%%%%%%%%%%%%%%%%%%%%%%%%%%%%%%%%%%%%%%%%%%%%%%%%%%%%%%%%%%%%%%%%%%

%%%%%%%%%%%%%%%%%%%%%%%%%%%%%%%%%%%%%%%%%%%%%%%%%%%%%%%%%%%%%%%%%%%%%%%%%%%%%%%%



% \section*{ACKNOWLEDGMENT}

% The preferred spelling of the word ÒacknowledgmentÓ in America is without an ÒeÓ after the ÒgÓ. Avoid the stilted expression, ÒOne of us (R. B. G.) thanks . . .Ó  Instead, try ÒR. B. G. thanksÓ. Put sponsor acknowledgments in the unnumbered footnote on the first page.

%%%%%%%%%%%%%%%%%%%%%%%%%%%%%%%%%%%%%%%%%%%%%%%%%%%%%%%%%%%%%%%%%%%%%%%%%%%%%%%%

% References are important to the reader; therefore, each citation must be complete and correct. If at all possible, references should be commonly available publications.

\bibliographystyle{IEEEtran}

\bibliography{IEEEabrv,IEEEexample}

\end{document}